\documentclass[12pt]{article}

\newcommand{\blind}{1}

\usepackage[inline]{enumitem}
\usepackage{float}
\usepackage{amsmath}
\usepackage{graphicx}
\usepackage{subcaption}
\usepackage[round]{natbib}
\usepackage{float}
\usepackage{geometry}
\usepackage{tikz}
\usepackage[english]{babel}
\usepackage{longtable}
\usepackage{color}
\usepackage{amssymb, amsmath, amsthm}
\usepackage{multirow}
\usepackage[titletoc,title]{appendix}
\usepackage{authblk}
\usepackage{setspace}
\usepackage{dsfont}
\RequirePackage[OT1]{fontenc}

\RequirePackage[colorlinks=true,citecolor=blue,urlcolor=blue,linkcolor=black]{hyperref}
\newtheoremstyle{indented}{5pt}{5pt}{\addtolength{\leftskip}{2.5em}}{}{\bfseries}{.}{.7em}{}

\usepackage{setspace,caption}
\usepackage{lipsum}
\newtheorem{lemma}{Lemma} 
\newtheorem{definition}{Definition}
\newtheorem{remark}{Remark}
\newtheorem{theorem}{Theorem}

\newcommand*\dd{\mathop{}\!\mathrm{d}}
\DeclareMathOperator{\expit}{expit}
\DeclareMathOperator{\logit}{logit}
\DeclareMathOperator{\var}{Var}

\DeclareMathOperator{\tmlee}{tmle}
\DeclareMathOperator{\etmlee}{\textit{e}-tmle-al}
\DeclareMathOperator{\etmleee}{\textit{e}-tmle}
\DeclareMathOperator{\ctmlee}{c-tmle}
\DeclareMathOperator{\ggcomp}{g-comp}
\DeclareMathOperator{\eff}{eff}
\DeclareMathOperator{\aipww}{aipw}

\pgfdeclarelayer{background}
\pgfsetlayers{background,main}
\newcommand{\Vertex}[2]
{\node[minimum width=0.6cm,inner sep=0.05cm] (#2) at (#1) {$\footnotesize#2$};
}
\newcommand{\Vertexr}[2]
{\node[rectangle, draw, minimum width=0.6cm,inner sep=0.05cm] (#2) at (#1) {$\footnotesize#2$};
}
\newcommand{\EdgeR}[3]%
{ \begin{pgfonlayer}{background}
    \draw[dashed,#3] (#1) to[bend right=30] (#2);
  \end{pgfonlayer}
}
\newcommand{\EdgeL}[3]%
{ \begin{pgfonlayer}{background}
    \draw[dashed,#3] (#1) to[bend left=30] (#2);
  \end{pgfonlayer}
}

\newcommand{\Arrow}[3]%
{ \begin{pgfonlayer}{background}
    \draw[->,#3] (#1) -- (#2);
  \end{pgfonlayer}
}

\newcommand{\indep}{\mbox{$\perp\!\!\!\perp\,$}}
\newcommand{\Pn}{\mathbb{P}_{n}} 
\newcommand{\PP}{\mathbb{P}} 
\renewenvironment{proof}{{\it Proof }}{\qed \\}
\newcommand{\tmle}{\hat\theta_{\tmlee}}
\newcommand{\etmle}{\hat\theta_{\etmlee}}
\newcommand{\etmlen}{\hat\theta_{\etmleee}}
\newcommand{\seff}{\sigma^2_{\eff}}
\newcommand{\setmle}{\sigma_{\etmlee}}
\newcommand{\hsetmle}{\hat{\sigma}_{\etmlee}}
\newcommand{\ctmle}{\hat\theta_{\ctmlee}}
\newcommand{\gcomp}{\hat\theta_{\ggcomp}}
\newcommand{\aipw}{\hat\theta_{\aipww}}

{
  \theoremstyle{definition}
  \newtheorem{assumption}{}
}
{
  \theoremstyle{definition}
  \newtheorem{assumptioniden}{}
}

\title{Doubly robust estimators for the average treatment effect under
  positivity violations: introducing the $e$-score} \if1\blind
\author[1]{Iv\'an D\'iaz \thanks{corresponding author:
    ild2005@med.cornell.edu}}
\affil[1]{\small Division of Biostatistics, Weill Cornell Medicine.}
\fi

\if0\blind
\author[1]{\vspace{-1.5cm} }
\fi

\allowdisplaybreaks
\begin{document}
\maketitle

\begin{abstract}
  Estimation of causal parameters from observational data requires
  complete confounder adjustment, as well as positivity of the
  propensity score for each treatment arm. There is often a trade-off
  between these two assumptions: confounding bias may be reduced
  through adjustment for a large number of pre-treatment covariates,
  but positivity is less likely in analyses with irrelevant predictors
  of treatment such as instrumental variables. Under empirical
  positivity violations, propensity score weights are highly variable,
  and doubly robust estimators suffer from high variance and large
  finite sample bias. To solve this problem, we introduce the
  $e$-score, which is defined through a dimension reduction for the
  propensity score
  . This dimension reduction is based on a result known as
  collaborative double robustness, which roughly states that a
  propensity score conditioning only on the bias of the outcome
  regression estimator is sufficient to attain double robustness. We
  propose methods to construct doubly robust estimators based on the
  $e$-score, and discuss their properties such as consistency,
  efficiency, and asymptotic distribution. This allows the
  construction of asymptotically valid Wald-type confidence intervals
  and hypothesis tests. We present an illustrative application on
  estimating the effect of smoking on bone mineral content in
  adolescent girls well as a synthetic data simulation illustrating
  the bias and variance reduction and asymptotic normality achieved by
  our proposed estimators.
\end{abstract}


\section{Introduction}

Estimation of causal effects from observational studies requires two
assumptions on the data generating mechanism: the assumption of
\textit{no unmeasured confounding}, and the assumption of
\textit{positivity} of the treatment probabilities. Positivity states
that individuals in all strata of the confounders have a positive
probability of getting assigned to each treatment arm
\citep{RosenbaumRubin83}. 
Theoretical positivity violations, whereby the \textit{true} treatment
probabilities are zero for some covariate strata, are problematic
because they preclude identification of the causal effect from
observational data. Empirical positivity violations, whereby the
\textit{estimated} treatment probabilities are close to zero for some
confounder strata, are also problematic because non-parametric regular
estimators of the causal effect suffer from large variability and
increased finite sample bias. Large amounts of pre-treatment data
poses a trade-off between these two
assumptions. \cite{vanderweele2011new} show that adjustment for more
pre-treatment variables reduces confounding bias, provided that all
adjustment variables are causes of either the treatment or the
outcome. However, instrumental variables, which we define as variables
that are cause of the treatment but are otherwise unrelated to the
outcome, are known to inflate the non-parametric efficiency bound
\citep{hahn2004functional,
  brookhart2006variable,greenland2008invited,schisterman2009overadjustment,rotnitzky2010note,myers2011effects},
and may lead to positivity violations.

In this article we focus on a class of estimators called \emph{doubly
  robust}. Double robustness is a property that ensures consistency of
the causal effect estimator under consistency of at least one of two
nuisance parameters: the outcome expectation conditional on treatment
and covariates (henceforth referred to as outcome regression), or the
probability of each treatment arm conditional on covariates
(henceforth referred to as the propensity score). Several doubly
robust methods for joint selection of the propensity score and outcome
regression models have been recently proposed
\citep{belloni2014inference,shortreed2017outcome,cheng2017estimating,ertefaie2018variable,koch2018covariate}. Generally,
these methods solve the trade-off between unconfoundedness and
positivity by performing variable selection for both models using
carefully constructed penalization terms in generalized linear
models. In this paper, we use the term \textit{high-dimensional data}
to mean a large number of variables that does not vary with sample
size. Though these parametric models may be useful with a few
variables, parametric assumptions in high-dimensional settings are
rarely justified by scientific knowledge
\citep{starmans2018predicament}. This implies that the models are
frequently misspecified, yielding inconsistent effect estimators
\citep[the consequences of parametric model misspecification in causal
inference were demonstrated in an influential simulation study
by][]{Kang2007}. Data-adaptive estimation methods offer an opportunity
to employ flexible estimators that are more likely to achieve
consistency. Methods such as those based on regression trees,
regularization, boosting, neural networks, support vector machines,
adaptive splines, and stacked ensembles of them, offer flexibility in
the specification of interactions, non-linear, and higher-order terms,
a flexibility that is not available for parametric models. Because of
this, machine learning has gained increasing popularity among causal
inference researchers \citep[e.g.,][]{vanderLaan&Petersen&Joffe05,
  Wang&Bembom&vanderLaan06,ridgeway2007, Bembometal08a,
  lee2010improving,neugebauer2016case}. Indeed, the framework of
\textit{targeted learning} \citep{vanderLaan&Rubin06,
  vanderLaanRose11,vanderLaanRose18}, concerned with the development
of $n^{1/2}$-consistent, asymptotically normal, and efficient
estimators of smooth low-dimensional parameters through the use
state-of-the art machine learning, has arisen as an alternative to the
widespread use of misspecified parametric models. Though much progress
has been made in targeted learning, joint model selection techniques
for causal inference using data-adaptive nuisance estimators remains
an open problem. Our manuscript aims to develop methodology to fill
this gap in the literature.

Our work is inspired by a result due to \cite{van2010collaborative},
called \textit{collaborative double robustness}, which roughly states
a propensity score adjusting for the bias of the outcome regression
estimator is sufficient to yield double robustness. Therefore, if the
outcome regression is consistent, no propensity score adjustment is
necessary, thus avoiding variance inflation and positivity
violations. Likewise, a propensity score adjusting only for the
point-wise bias of the outcome regression, suffices for
consistency. This result was used in a series of papers to develop a
number of estimators collectively known as \textit{collaborative
  targeted minimum loss based estimators}
\citep[C-TMLE,][]{van2010collaborative,susan2010application,ju2017scalable,
  ju2017collaborative, ju2018collaborative}. These instances of C-TMLE
solve the trade-off between unconfoundedness and positivity by
introducing joint model selection techniques for the outcome
regression and propensity score. They can be described as model
selection techniques for the propensity score that optimize a suitably
constructed loss function which takes into account the outcome
regression bias. For example, the original C-TMLE was developed as a
variable selection tool using a greedy search 
. The refinements of \citeauthor{ju2018collaborative}
extended C-TMLE to more general model selection frameworks with
continuously indexed candidate estimators for the propensity score
such as $\ell_1$ regularization \citep{ju2017collaborative}. The model
selection approaches employed by existing C-TMLE methods have a time
complexity that, in the best case scenario, grows linearly with the
dimension of the adjustment vector. This time complexity may be
computationally prohibitive in certain big data settings. Furthermore,
it is not clear how these model selection approaches can be
generalized to general data-adaptive estimators, for example
tree-based approaches, support vector machines, neural networks, or
learning ensembles.

Our main contribution and innovation is to present an alternative
collaborative double robustness result, whereby we reduce the
dimension of the propensity score through what we define as the
$e$-score. 
The $e$-score and its double robustness property allows us to propose
estimation methods that do not involve complex model selection
algorithms and are therefore completely scalable as well as
generalizable to any initial data-adaptive estimator of the propensity
score. Our second main contribution is to study the asymptotic
distributions of the proposed collaborative estimator under consistent
estimation (and convergence rates) of both nuisance parameters. This
asymptotic result is fundamental to the construction of valid
confidence intervals and hypothesis tests.

\section{Notation and Inferential Problem}\label{sec:problem}
Let $W$ denote a vector of observed baseline variables, let $A$ denote
a treatment indicator, and let $Y$ denote the outcome of
interest. Throughout, we assume that $Y$ takes values on $[0, 1]$.
The word \textit{model} here refers to a set of probability
distributions for the observed data $O=(W, A, Y)$. We assume that the
true distribution of $O$, denoted by $\PP$, is an element of the
nonparametric model, denoted by $\cal M$, and defined as the set of
all distributions of $O$ dominated by a measure of interest $\nu$.
Assume we observe an i.i.d. sample $O_1,\ldots,O_n$, and denote its
empirical distribution by $\Pn $. For a general distribution $P$ and a
function $f$, we use $Pf$ to denote $\int f(o)dP(o)$.

Let $Y_a:a\in\{0,1\}$ denote the counterfactual outcome that would be
observed in a hypothetical world in which $P(A=a)=1$. The typical
observational study is focused on estimation of the counterfactual
expectations $E(Y_a)$, or contrasts between them. We focus on
estimating $E(Y_1)$; estimators of $E(Y_0)$ may be constructed using
symmetric arguments. We use $m(w)$ to denote the outcome regression
$E_{\PP}(Y\mid A=1,W=w)$, $g(w)$ to denote the propensity score
$\PP(A=1\mid W=w)$.



We introduce the following
assumptions, which are standard in the causal inference literature.

\begin{assumptioniden}[No unmeasured confounders]\label{ass:random}
  $A$ is independent of $Y_1$ conditional on $W$.
\end{assumptioniden}
\begin{assumptioniden}[Strong positivity]\label{ass:positivity}
  $\PP\{g(W) > \epsilon\} = 1$ for some $\epsilon>0$.
\end{assumptioniden}

Assumption \ref{ass:random} states that treatment assignment is
randomized within strata of the covariates, either by nature or by
experimentation. We make assumptions \ref{ass:random} and
\ref{ass:positivity} throughout the manuscript. The mean
counterfactual outcome $E(Y_1)$ is identified from the distribution
$\PP$ of the observed data as $\theta = E_{\PP}\{m(W)\}$ \cite[see
e.g.,][]{Pearl00}. We define the target parameter mapping as
$\theta(\PP) = E_{\PP}\{m(W)\}$.

\subsection{Existing estimators and asymptotic
  properties}\label{sec:existing}

Doubly robust and efficient estimation of $\theta$ in the
non-parametric model proceeds as follows. Define the estimating
function
\begin{equation}
  D_{\eta, \theta}(O)=\frac{A}{g(W)}\{Y-m(W)\} + m(W) - \theta,\label{eq:defDeta}
\end{equation}
where $\eta=(g,m)$. The estimating function $D_{\eta,\theta}(O)$ is
a fundamental object for the construction of estimators of $\theta$
in the non-parametric model. On one hand, $D_{\eta,\theta}(O)$
characterizes the efficiency bound in the sense that all regular
estimators have a variance that is larger or equal to
$\seff=\var\{D_{\eta,\theta}(O)\}$ \citep{hahn1998role}. On the
other hand, for an estimate $\hat\eta$ of $\eta$, any estimator
$\hat \theta$ which is a solution of the estimating equation
$\Pn D_{\hat\eta, \theta}= 0$ on $\theta$ is doubly robust, meaning
that it is consistent if at least one of $\hat g$ and $\hat m$ is
consistent \citep[see Theorem 5.9 in][]{vanderVaart98}. Double
robustness follows from the fact that $\PP D_{\eta_1,\theta} = 0$ if
either $g_1=g$ or $m_1=m$, where $\eta_1$ denotes the limit of
$\hat\eta$ as $n\to \infty$.

The estimator obtained by directly solving the estimating equation
$\Pn D_{\hat\eta, \theta}= 0$ is also called the \textit{augmented
  inverse probability weighted estimator}, and we denote it with
$\aipw$. This estimator is often critiqued because it can lead to
estimates outside of the parameter space \citep{Gruber2010t}. Several
estimators have been proposed to remedy this issue \cite[see
e.g.,][]{Kang2007, Robins2007, tan2010bounded}. In this paper we focus
on the \textit{targeted minimum loss based estimation} (TMLE)
methodology, developed by \cite{vanderLaan&Rubin06}. We now briefly
review the construction of a TMLE. Further discussion on the
construction of the TMLE for this problem may be found in
\cite{Gruber2010t}. \cite{Porter2011} provides an excellent review of
other doubly robust estimators along with a discussion of their
strengths and weaknesses.

The TMLE of $\theta$ is defined as $\tmle=\theta(\tilde P)$, where
$\tilde P$ is an estimator of $\PP$ constructed to satisfy
$\Pn D_{\tilde \eta, \tmle}=0$. The estimator $\tilde P$ is
constructed by tilting an initial estimate $\hat P$ towards a solution
of the estimating equation, by means of parametric
submodel. Specifically, a TMLE may be constructed by fitting the
logistic regression model
\begin{equation}
  \logit  m_\beta(w) = \logit \hat m(w) + \beta \frac{1}{\hat
    g(w)},\label{eq:submodel}
\end{equation}
among observations with $A=1$. Here,
$\logit(p)=\log\{p(1-p)^{-1}\}$. In this expression $\beta$ is the
parameter of the model, $\logit \hat m(w)$ is an offset variable, and
the initial estimates $\hat m$ and $\hat g$ are treated as fixed. The
parameter $\beta$ is estimated through the empirical risk minimizer
\[\hat \beta = \arg\max_{\beta}\sum_{i=1}^n A_i \{Y_i\log
  m_\beta(W_i) + (1-Y_i)\log(1- m_\beta(W_i))\}.\] The tilted
estimator of $m(w)$ is defined as
$\tilde m(w) = m_{\hat \beta}(w) = \expit\{\logit \hat m(w)+
\hat\beta / \hat g(w)\}$, where $\expit(x) = \logit^{-1}(x)$. The
TMLE of $\theta$ is defined as
\[\tmle=\frac{1}{n}\sum_{i=1}^n \tilde  m(W_i).\]
Because the empirical risk minimizer of model (\ref{eq:submodel})
solves the score equation
\[\sum_{i=1}^n\frac{A_i}{\hat g(W_i)}\{Y_i -  m_{\hat
    \beta}(W_i)\}=0,\] it follows that $\Pn D_{\tilde \eta, \tmle}=0$
with $\tilde \eta = (\hat g, \tilde m)$. The analysis of the
asymptotic properties of the TMLE and other estimators that solve the
estimating equation $\Pn D_{\hat\eta,\hat\theta}=0$ may be based on
standard $M$-estimation and empirical process theory. 
In particular, under regularity conditions including convergence of
$\hat g$ and $\hat m$ at rates at least as fast as $n^{-1/4}$, it may
be shown that $\tmle$ is asymptotically linear \citep[see e.g.,][for
all technical details]{vanderLaanRose11}:
\begin{equation*}
  \tmle - \theta=(\Pn - \PP)D_{\eta, \theta} + o_P\big(n^{-1/2}\big).\label{eq:tmleef}
\end{equation*}
Together with the above result, the CLT shows that $\tmle$ is
efficient in the sense that its asymptotic variance is equal to the
efficiency bound
\begin{equation}
  \seff=\var\{D_{\eta,\theta}(O)\}=E\left\{\frac{\sigma_0^2(W)}{g(W)}\right\} +
  E\{m(W)-\theta\}^2,\label{eq:effb}
\end{equation}
where $\sigma_0^2(w)=\var(Y\mid A=1, W=w)$. Inspection of this bound
reveals which variables must be selected in order to improve
precision. First, the conditional variance $\sigma_0^2(w)$ must be
small, which implies that all predictors of the outcome must be
included in the outcome regression, regardless of whether they are
necessary for confounder adjustment. Second, the propensity score
$g(w)$ must be bounded away from zero, which means that variables
that are predictors of $A$, but are unnecessary for confounder
adjustment, must be excluded \citep{hahn2004functional}.

The estimator we propose to solve this problem is closely related to
the \textit{collaborative targeted minimum loss based estimator}
(C-TMLE) proposed by \cite{van2010collaborative}. C-TMLE is built upon
a property known as \textit{collaborative double robustness}, defined
as follows. To introduce collaborative double robustness, we will
require the following assumption:
\begin{assumption}[Doubly robust consistency]\label{ass:DR1}
  Let $||\cdot||$ denote the $L_2(\PP)$ norm defined as
  $||f||^2=\int f^2 \dd\PP$. Assume there exists $\eta_1=(g_1, m_1)$
  with either $g_1 = g$ or $ m_1= m$ such that
  $||\hat m - m_1||=o_P(1)$ and $||\hat g - g_1||=o_P(1)$.
\end{assumption}
We reproduce the original theorem \citep[Theorem 2
of][]{van2010collaborative}:
\begin{theorem}[Collaborative double robustness]\label{theo:cdr1}
  Let $s(w)$ denote the asymptotic pointwise bias in estimation of
  $m(w)$. That is, define $s(w)=m(w) - m_1(w)$. Let
  $g_s(w) = \PP(A=1\mid s(W) = s(w))$. Assume
  $\eta_1=(g_1, m_1)$ is such that either $g_1 = g_s$, or
  $m_1=m$. Then $\PP D_{\eta_1, \theta} = 0$.
\end{theorem}
The above theorem implies that the probability $g(W)$ does not need to
adjust for the full covariate vector $W$. A propensity score $g_s(W)$
that only adjusts for the residual error $s(W)$ is sufficient to
obtain a doubly robust estimating equation. This dimension reduced
propensity score has lower or equal variance to the original
propensity score. In particular, since
$g_s(w) = E_{\PP}\{g(W)\mid s(W) = s(w))$, the law of total variance
yields
$\var\{g(W)\} = E\left\{\var[g(W)\mid s(W)]\right\} + \var\{g_s(W)\}$,
which implies $\var\{g(W)\} \geq \var\{g_s(W)\}$.  Thus, usage of
$g_s(w)$ instead of $g(w)$ can generate efficiency gains in estimation
of $\theta$. Though this result is more general, a particular instance
in which it is helpful is in the presence of instrumental
variables. If the estimator $\hat m$ is misspecified but correct in
the sense that $m_1$ does not depend on the instruments, then the
propensity score does not need to adjust for the instruments,
irrespective of their correlation with $A$. This formalizes the advice
of \cite{brookhart2006variable} and others in the sense that only
variables related the outcome should be included in the propensity
score model. The original C-TMLE algorithm operates under a sparsity
assumption that the residual bias $s(W)$ is a function of a subset of
the covariates $W$, and proceeds by constructing clever variable
selection algorithms to find such subset. In the following section we
introduce the $e$-score, which is inspired in the collaborative double
robustness result of Theorem~\ref{theo:cdr1}. Unlike the C-TMLE, the
$e$-score reduces the variance of the propensity score without
sparsity assumptions, therefore allowing us to construct methods
applicable to general data-adaptive estimators of the propensity score
such as those based on machine or statistical learning.

\section{Collaborative double robustness based on the
  $e$-score}\label{sec:escore}

We start this section by presenting an alternative collaborative
double robustness theorem, which provides the foundation for our
proposed estimator. Our result is based on the collaborative double
robustness principle that, when the outcome regression is consistently
estimated at the appropriate rate, then the propensity score may be
simply defined as $\PP(A=1)$. More generally, a propensity score that
adjusts for the asymptotic bias of the outcome regression estimator
suffices to attain double robustness (Theorem~\ref{theo:cdr1}). 

\begin{definition}[$e$-score]
  Assume $g_1$ and $m_1$ are as in \ref{ass:DR1}. Let
  \[r_1(w) = E\big\{Y - m_1(W)\mid A=1,g_1(W) = g_1(w)\big\}\]
  quantify the outcome model misspecification as a function of the
  possibly misspecified limit of the propensity score estimator. The
  $e$-score is defined as
  \[  e_1(w) = E\left\{g_1(W)\mid r_1(W) = r_1(w)\right\}.\]
\end{definition}
Theorem~\ref{theo:cdr2}, stated rigorously below, teaches us that an
estimator based on the efficient influence function, but constructed
using $e_1$ instead of $g_1$, maintains the double robustness
property. To introduce this result, define the estimating function
\begin{equation}
  D_{\lambda,\theta}(O) = \frac{A}{e(W)}\{Y-m(W)\} + m(W) - \theta,\label{eq:defDlambda}
\end{equation}
where we have denoted $\lambda = (e, m)$.
\begin{theorem}[Double robustness based on the $e$-score]\label{theo:cdr2}
  Let $(g_1,m_1)$ be such that either $g_1=g$ or $m_1=m$. Let
  $\lambda_1 = (e_1, m_1)$. Then $\PP D_{\lambda_1,\theta} = 0$.
\end{theorem}
We note that this result is different in nature from both standard and
collaborative double robustness. The first sense in which
$D_{\lambda,\theta}$ is robust is similar to standard double
robustness: if the outcome regression is correctly specified then the
propensity score may be replaced by the $e$-score, which can be any
function $e_1:w\mapsto (0,1)$. The second way in which
$D_{\lambda,\theta}$ is robust is novel: if the outcome regression is
misspecified, the propensity score may be replaced by the $e$-score,
provided that the propensity score is consistently estimated. In
comparison to collaborative double robustness
(Theorem~\ref{theo:cdr1}), the result in Theorem~\ref{theo:cdr2} is
about consistent estimation of the propensity score that conditions on
the full vector $W$, as opposed to the reduced-data propensity score
required in Theorem~\ref{theo:cdr1}.

The main advantage of the $e$-score in comparison to the propensity
score is the reduction of the variability of the weights by only
adjusting for the residual bias, as measured by $r_1(w)$. In
particular, if $m_1=m$, then the law of iterated expectation shows
that $r_1(w)=0$, and the $e$-score is a constant equal to the constant
$E[g_1(W)]$. If $g_1=g$, the $e$-score reduces the variance of $g$
through adjustment for the outcome residual bias as quantified by
$r_1(w)$.  To further illustrate this, consider a partition of
$W=(W_I,W_P)$. Assume that, unknown to the researcher, the causal
structure of the variables is as depicted in the directed acyclic
graph of Figure~\ref{fig:dag}. The fact that the relationship between
$A$ and $Y$ is unconfounded is not known to the researcher, so she
decides to adjust for the full vector $W$. This unnecessarily
increases the efficiency bound of the model and the variance of doubly
robust estimators. Usage of the $e$-score fixes this problem as
follows. Assume that the estimator $\hat m$ is inconsistent but
sensible in the sense that $m_1$ only depends on $w_P$. Since $g$ only
depends on $w_I$, and $W_I\indep W_P$, we have $r_1(w)$ is a constant
equal to $E[m(W_P)-m_1(W_P)]$, and the $e$-score is equal to
$\PP(A=1)$, therefore recovering the efficiency bound of a model in
which the causal structure of Figure~\ref{fig:dag} is known.

\begin{figure}[!htb]
  \centering
  \begin{tikzpicture}
    \Vertex{0, 0}{A}
    \Vertex{2, 0}{Y}
    \Vertex{-2, 0}{W_I}
    \Vertex{4, 0}{W_P}
    \Arrow{W_I}{A}{black}
    \Arrow{W_P}{Y}{black}
    \Arrow{A}{Y}{black}
  \end{tikzpicture}
  \caption{Directed acyclic graph.}
  \label{fig:dag}
\end{figure}


If the outcome model misspecification is such that the residual
$Y - m_1(W)$ is a monotone function of $g_1(W)$, then we have
$e_1(W)=g_1(W)$. In this case our collaborative doubly robustness
reduces to standard double robustness. To avoid this pathological
case, we recommend to explicitly include $\hat g(W)$ as a covariate
when computing the estimator $\hat m(W)$.

We now proceed to discuss several alternatives to construct a
collaborative doubly robust estimator based on the $e$-score. 



\section{Proposed Estimators}\label{sec:proposal}

In this section we propose two estimators for $\theta$ based on the
collaborative double robustness result of
Theorem~\ref{theo:cdr2}. Both estimators are constructed under the
targeted minimum loss based framework. The first estimator is purely
based on obtaining a tilted estimator $\tilde m$, which targets a
solution to an estimating equation based on $D_{\lambda, \theta}$. A
large sample analysis of this estimator reveals that it is likely not
asymptotically linear in many important situations. As a solution to
this flaw, we propose a second estimator, in which we target
additional estimating equations that yield asymptotic linearity.

To start, we discuss estimators of $r_1(w)$ and $e_1(w)$. Note that these
quantities are one-dimensional regression functions which can be
consistently estimated using non-parametric estimators, e.g., kernel
smoothing, smoothing splines, the highly adaptive lasso, etc.. For
example, a for a second-order kernel function $K_h$ with bandwidth $h$
a kernel estimator of $r_1(w)$ would be defined as
\[  \hat r(w) = \frac{\sum_{i = 1}^nA_i\,K_{\hat h}\{\hat g(W_i) -
    \hat g(w)\}\{Y_i - \hat m(W_i)\}}{\sum_{i=1}^nA_i\,K_{\hat h}\{\hat g(W_i) -
    \hat g(w)\}},
\]
and a kernel estimator of $e_1(w)$ would be defined analogously. Once
$\hat e(w)$ is computed, a variance-reduced TMLE can be computed by
applying the TMLE algorithm presented in Section~\ref{sec:existing}
with $\hat g(w)$ replaced by $\hat e(w)$. Denote such estimator with
$\etmlen$. The analysis of the asymptotic properties of $\etmlen$
follows standard arguments in the analysis of $M$-estimators, as in
Section~\ref{sec:existing}. Define the following Donsker condition:
\begin{assumption}[Donsker]\label{ass:donsker}
  Let $\eta_1$ be as in~\ref{ass:DR1}. Assume the class of
  functions $\{\lambda=(e, m):||m - m_1||<\delta, ||e-e_1||<\delta\}$ is
  Donsker for some $\delta >0$.
\end{assumption}
Under \ref{ass:DR1} and \ref{ass:donsker}, a straightforward
application of Theorems 5.9 and 5.31 of \cite{vanderVaart98} \citep[see also
example 2.10.10 of][]{vanderVaart&Wellner96} yields
\begin{equation}
  \etmlen-\theta= \beta(\hat\lambda) +
  (\Pn - \PP)D_{\lambda_1, \theta} + o_P\big(n^{-1/2} + |\beta(\hat\lambda)|\big),\label{eq:wh}
\end{equation}
where $\beta(\hat\lambda) = \PP D_{\hat\lambda, \theta}$. From
equation (\ref{eq:wh}) we can see that the only missing element to
understand the asymptotic distribution of $\etmlen$ is the ``drift''
term $\beta(\hat\lambda)$. If this term, which is equal to
\begin{equation}
  \beta(\hat\lambda) = \int
  \frac{1}{\hat e}(g - \hat e)(m - \hat m)\dd\PP,\label{eq:defbeta}
\end{equation}
can be shown to be asymptotically linear in the sense that
\begin{equation}
\beta(\hat\lambda)=(\Pn-\PP)S + o_P(n^{-1/2}),\label{eq:aslinb}
\end{equation}
for some function $S$ of $O$ that may depend on $\PP$, then asymptotic
linearity and normality of $\etmlen$ follows. Unfortunately,
$\beta(\hat\lambda)$ is a complex term that cannot be expected to
satisfy (\ref{eq:aslinb}) in general. Recall that $\hat m$ and
$\hat g$ are constructed using general data-adaptive methods, with the
only constraints that the estimators must satisfy conditions
\ref{ass:DR1} and \ref{ass:donsker}. These conditions are satisfied
for a large number of estimators for which (\ref{eq:aslinb}) does not
hold. See for example \citep{bickel2009simultaneous} for rate results
on $\ell_1$ regularization, \citep{wager2015adaptive} for rate results
on regression trees, and \citep{chen1999improved} for neural
networks. These conditions are also satisfied by the highly adaptive
lasso \citep{benkeser2016highly} under the assumption that the true
regression function is right-hand continuous with left-hand limits and
has variation norm bounded by a constant. Although all of these
methods satisfy \ref{ass:DR1} and \ref{ass:donsker}, they do not
generally satisfy ({\ref{eq:aslinb}).

\subsection{Achieving asymptotic linearity}

We now propose a second estimator, $\etmle$, which is asymptotically
linear. Asymptotic linearity is important because it implies
asymptotic normality and facilitates the construction of confidence
intervals and hypothesis tests. This goal will be achieved under the
assumption of consistent estimation of both nuisance parameters, $g$
and $m$. Our estimator $\hat\lambda$ guarantees the asymptotic
linearity of $\beta(\hat\lambda)$, under certain conditions, by
tilting the initial estimator towards a solution of a score equation
carefully constructed to target $\beta(\hat\lambda)$. To develop this
construction, we start by requiring specific convergence rates for all
nuisance estimators:
\begin{assumption}[Consistency rate of nuisance estimators]\label{ass:DR2}
  Assume $||\hat m - m||\,||\hat g - g||=o_P(n^{-1/2})$.
\end{assumption}
Under the above condition, which is standard in the analysis of doubly
robust estimators, the following lemma provides a representation of
the drift term in terms of score functions. This representation is
achieved through the following univariate regression:
\begin{equation*}
  q(w) = E\left\{A\left(\frac{1}{\hat e(W)} -
      \frac{1}{\hat g(W)}\right)\,\bigg|\,\, \hat m(W) = \hat m(w)\right\},
\end{equation*}
where the expectation is taken with respect to the distribution of
$(A,W)$, taking $\hat e$, $\hat g$, and $\hat m$ as fixed functions.
Like $r_1$ and $e_1$, we estimate $q$ consistently through
non-parametric univariate regression methods.
We have the following result:
\begin{lemma}[Asymptotic representation of the drift term]\label{lemma:driftrep}
  Let $h(w) = q(w)/g(w)$, and define the score function
  \[ S_h(O) = A\, h(W)\{Y - \hat m(W)\}\] Under
  \ref{ass:DR2} we have $\beta(\hat\lambda) = \PP S_{h} + o_P(n^{-1/2})$.
\end{lemma}
The proof of the lemma is presented along with all other proofs in the
Supplementary Materials. The above lemma sheds light on the necessary
characteristics of an estimator $\hat\lambda$ in order to satisfy
(\ref{eq:aslinb}). In particular, the proof of Theorem~\ref{theo:dr}
shows that asymptotic linearity of $\hat \theta$ requires that
$\hat\lambda$ solve the score equation $\Pn S_{\hat h}=0$ for
$\hat h(w) = \hat q(w)/\hat g(w)$.

We now describe in detail our proposed estimator, which we denote
$\etmle$, and define through the following iterative algorithm.
\begin{enumerate}[label = Step~\arabic*., align=left, leftmargin=*]
\item \textit{Initial estimators.} Obtain initial estimators $\hat g$
  and $\hat m$ of $g$ and $m$. Construct estimators $\hat r$,
  $\hat e$, and $\hat h$ using kernel regression estimators
  as described above.
\item \textit{Solve estimating equations.} Estimate the parameter
  $\beta = (\beta_1, \beta_2)$ in the logistic tilting model
  \begin{align}
    \logit  m_\beta(w) &= \logit \hat m(w)  + \beta_{1} {\hat e(w)}^{-1} +
                            \beta_{2} \hat h(w),\label{eq:submodelY}
  \end{align}
  Here, $\logit \hat m(w)$ is an offset variable (i.e., a variable
  with known parameter equal to one). The parameter
  $(\beta_{1}, \beta_{2})$ may be estimated through a logistic
  regression model of $Y$ on the bivariate
  vector $[{\hat e(W)}^{-1}, \hat h(W)]$, with no intercept and with
  offset $\logit \hat m(W)$ among observations with $A=1$. Let
  $\hat\beta$ denote these estimates.
\item \textit{Update estimator and compute $e$-TMLE.} Define the updated
  estimator as $\tilde m =  m_{\hat \beta}$. The proposed TMLE of
  $\theta$ is defined as
  \[\etmle = \frac{1}{n}\sum_{i=1}^n \tilde m(W_i).\]
\end{enumerate}

In order to prove the asymptotic linearity of $\etmle$, we require the
following additional consistency assumption on the nuisance estimators
\begin{assumption}[Consistency rate of $(m,q)$ nuisance
  estimators]\label{ass:DR3}
  Assume $||\hat m - m||\,||\hat q - q||=o_P(n^{-1/2})$.
\end{assumption}

A sufficient assumption for the above condition to hold is that
$||\hat m - m||=o_P(n^{-1/4})$ and $||\hat q - q||=o_P(n^{-1/4})$.
Note that the latter convergence is purely about the consistency of
the smoothing method used to obtain $\hat q$, because the covariates
$\hat m$, $\hat g$, and $\hat e$ are the same in $\hat q$ and
$q$. Non-parametric smoothing methods can be expected to satisfy this
assumption in certain situations. For example, under the assumption
that the map
$x \mapsto E\left\{A\left(\frac{1}{\hat e(W)} - \frac{1}{\hat
      g(W)}\right)\,\bigg|\,\, \hat m(W) = x\right\}$ is twice
differentiable, a kernel regression estimator with optimal bandwidth
guarantees the desired convergence rate $||\hat
q-q||=o_P(n^{-1/4})$. The HAL also achieves the desired rate under the
assumption that the map is c\`adl\`ag with bounded sectional variation
norm \citep{benkeser2016highly}.

The large sample distribution of the above TMLE is given in the
following theorem:

\begin{theorem}[Asymptotic Linearity of $\etmle$]\label{theo:dr}
  Assume \ref{ass:donsker}, \ref{ass:DR2}, and \ref{ass:DR3}
  hold. Define $g_m(w) = P(A=1\mid m(W) = m(w))$. Then
  \[\etmle - \theta=(\Pn - \PP)\mathrm{IF}_0 +
    o_P\big(n^{-1/2}\big),\]
  where
\begin{equation}
  \mathrm{IF}_0(O) = \frac{A}{g(W)}\left\{1-\frac{g_m(W) -
      g(W)}{\PP(A=1)}\right\}\{Y-m(W)\}+m(W)-\theta.\label{eq:if}
\end{equation}
\end{theorem}
The proof of the above theorem is presented in the Supplementary
Materials. Together with the central limit theorem,
Theorem~\ref{theo:dr} shows that
$n^{1/2}(\etmle - \theta)\to N(0, \setmle^2)$, where $\setmle^2 = \var\{\mathrm{IF}_0(O)\}$.
The asymptotic distribution of Theorem~\ref{theo:dr} may be used to
construct hypothesis tests and a Wald-type confidence interval as
follows. In particular, the standard error may be estimated as
follows. Under Condition \ref{ass:DR1}, we have
$\hat e(W)\to \PP(A=1)$, and therefore
\[\hat q(w)\to E_{\PP}\left\{A\left(\frac{1}{\PP(A=1)} -
      \frac{1}{g(W)}\right)\,\bigg|\,\, m(W) = m(w)\right\}\] The law
of total expectation shows that the right hand side of the above
expression is equal to $g_m(w)/\PP(A=1) - 1.$ 

Thus, since $\hat e(w)$ is a consistent estimator of $\PP(A=1)$, and
$\hat q(w)$ is a consistent estimator of $g_m(w)/p - 1$,
$1/\hat e(W_i) - \hat q(W_i)/\hat g(W_i)$ is a
consistent estimator of $1-\{g_m(W) - g(W)\}/\PP(A=1)$, and
\[\hsetmle^2 = \frac{1}{n}\sum_{i=1}^n\left[A_i\left\{\frac{1}{\hat e(W_i)} -
      \frac{\hat q(W_i)}{\hat g(W_i)}\right\}\{Y_i - \tilde m(W_i)\} +
    \hat m(W_i) - \etmle\right]^2\] is a consistent estimator of
$\setmle^2$. Thus, the interval
$\etmle \pm z_{\alpha/2} n^{-1/2} \hsetmle$ has correct asymptotic
coverage $(1-\alpha)100\%$, whenever $\hat g$ and $\tilde m$ converge
to their true value at the rate stated in \ref{ass:DR2}.

Surprisingly, the above theorem does not require $\hat e$ to converge
to $e_1$ at any specific rate. On the other hand, the theorem does
require convergence of $\hat q$ at a rate that depends on a second
order term (assumption \ref{ass:DR3}). This sheds light on the role of
these two nuisance estimators: $\hat e$ is used to achieve double
robustness, whereas $\hat q$ is used to endow the estimator with a
valid asymptotic normal distribution.

\begin{remark}
  Inspection of equation (\ref{eq:if}) reveals the intuition behind
  the expected efficiency gains. If $\epsilon <g(W)\leq \PP(A=1)$, then by
  construction we have $g(W)\leq g_m(W)\leq \PP(A=1)$. Thus, large
  inverse probability weights $\{g(W)\}^{-1}$ get shrunk by a factor
  $0 \leq 1 - \{g_m(W) - g(W)\}/\PP(A=1) \leq 1$. When there are many
  large weights, this shrinkage has the effect of reducing the
  variance of the estimator $\etmle$ in comparison to the efficiency
  bound $\seff$ defined in (\ref{eq:effb}). To illustrate this
  further, consider an extreme scenario where $m(W)$ is independent
  of $A$ such that $g_m(W) = \PP(A=1)$. Then $\setmle^2$ reduces to
  $E[\PP(A=1)^{-2}\sigma^2(W)g(W) + \{m(W) - \theta\}^2]$. If the
  correlation between $\sigma^2(W)$ and $g(W)$ is small enough, then
  it can be expected that $\setmle^2 < \seff$. The stabilization of
  large probability weights comes at the price of larger weights for
  observations $W$ with large probabilities $g(W)> \PP(A=1)$. However, the
  weight augmentation factor is bounded by 2. In pathological cases
  where the correlation between $\sigma^2(W)$ and $g(W)$ is large,
  so that $\sigma^2(W)/g(W)$ is nearly constant, then it is possible
  that $\setmle^2 > \seff$. Equality of $\mathrm{IF}_0$ with the
  efficient influence function $D_{\eta,\theta}$ is obtained
  trivially when $g_m(w) = g(w)$, in which case $\etmle$ and
  $\tmle$ are asymptotically equivalent.
\end{remark}

\begin{remark}
  Our estimator is related to a recent proposal by
  \cite{benkeser2019nonparametric} which consists of replacing the
  propensity score by $g_m(w)$. Their proposed estimator
  requires consistent estimation of the outcome regression and is
  always super-efficient. Unlike their estimator, $\etmle$ is doubly
  robust and not uniformly super-efficient.
\end{remark}

\section{Simulation Studies}\label{sec:simula}

In this section we present a simulation study using synthetic data
with the aim of illustrating the properties of the proposed
estimators, in comparison with $\ctmle$, $\tmle$, $\aipw$, and the
G-computation estimator $\gcomp$. For each sample size 200, 800, 1800,
3200, 5000, 7200, 9800, 12800, we generate $1000$ datasets as
follows. First, a set of variables $\{Z_1,\ldots,Z_{15}\}$ is
generated, where all $Z_i$'s are independently distributed
$2\mathrm{Beta}(1/3,1/3) - 1$. Then, a set of covariates
$\{W_1,\ldots W_{15}\}$ is generated as $W_j=Z_j$ for odd $j$ and
$W_j = Z_{j-1}Z_j$ for even $j$. Then, a variable $A$ is drawn from a
Bernoulli distribution with probabilities
$g(W) = \expit\{\delta \sum_{j=1}^{10}W_j\}$, for
$\delta\in\{0,1\}$. The case $\delta = 0$ is a randomized trial and
represents a best-case scenario for the variability of the propensity
score. Figure~\ref{fig:gdens} shows the high variability of the
propensity score for $\delta = 1$. The outcome is generated as
$Y = A + \sum_{j=6}^{15} W_j + N(0,1)$. We aim to estimate the causal
effect of $A$ on $Y$, defined as $E(Y_1-Y_0)$. The efficiency bounds
for this parameter are approximately $6.8$ and $16.1$ for $\delta =0$
and $\delta = 1$, respectively. Note that only $W_6,\ldots,W_{10}$ are
confounders of the causal effect of $A$ on $Y$. Note also that the
causal effect of $A$ on $Y$ is $\theta=1$.

\begin{figure}[!htb]
  \begin{center}
    \includegraphics[scale = 0.3]{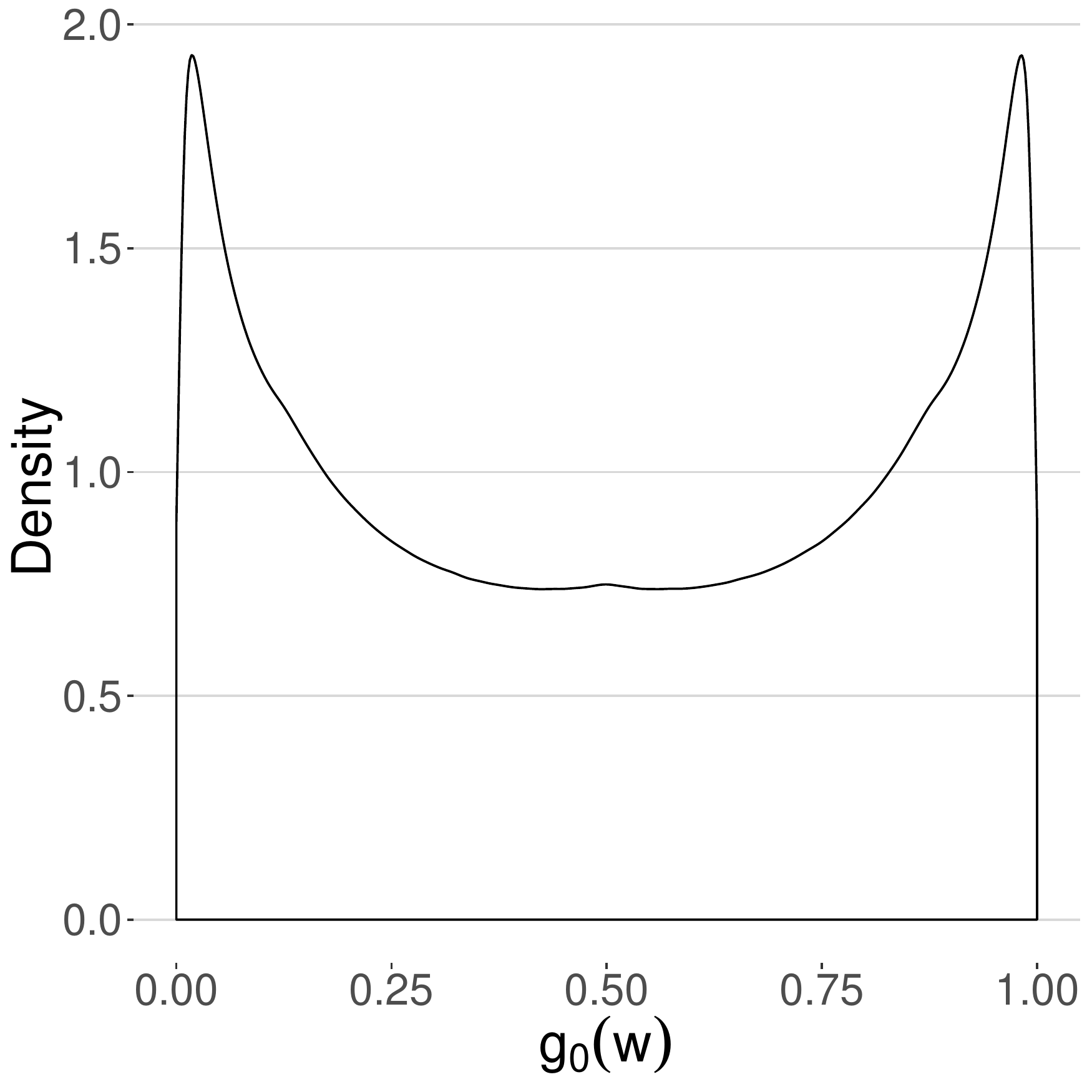}
  \end{center}
  \caption{Probability density function of the propensity score
    $g(W)$ in the simulation study.}
  \label{fig:gdens}
\end{figure}

For each generated dataset, we fit four different scenarios of
consistent estimation of $g$ and $m$: (A) both consistently
estimated, (B) only $m$ consistently estimated, (C) only $g$
consistently estimated, and (D) both inconsistently estimated. All
models consisted of main terms generalized linear regression models
with the appropriate link functions (identity for the outcome,
logistic for the propensity score). For $\hat g$, we fit a logistic
regression model. For $\hat m$, we fit a generalized additive model
that includes $\hat g(W)$ as a covariate. Consistent estimators were
constructed using covariates $W_j$; inconsistent estimators used
covariates $Z_j$. For each of the above scenarios, we computed the
four estimators: the $G$-computation or regression adjusted estimator,
$\etmle$, $\ctmle$, $\tmle$. We compare the performance of the
estimators in terms of four metrics:
\begin{itemize}
\item Absolute bias: $| E(\hat\theta - \theta)|$
\item Absolute bias scaled by $n^{1/2}$: $n^{1/2}| E(\hat\theta
  - \theta)|$
\item Standard deviation scaled by $n^{1/2}$:
  $n^{1/2} \text{sd}(\hat\theta)$
\item Root mean squared error scaled by $n^{1/2}$:
  $\{n E(\hat\theta - \theta)^2\}^{1/2}$.
\item The quotient $\hat \sigma^2 / \var(\hat\theta)$. For $\ctmle$ and
  $\tmle$, the variance was estimated using the variance of the
  efficient influence function. For $\etmle$, the variance was
  estimated using the doubly robust asymptotic distribution given in
  Theorem~\ref{theo:dr}. The $G$-computation estimator is not included
  in this comparison.
\item Coverage probability of a Wald-type confidence interval.
\end{itemize}
All the above quantities were approximated using Monte-Carlo integrals
across the 1000 generated datasets. The results for $\delta=1$ are
presented in Figures~\ref{fig:simula1} and~\ref{fig:simula2}. The
results for $\delta=0$, presented in the Supplementary Materials,
corroborate that all estimators have nearly identical performance
except in small samples.

\begin{figure}[!htb]
  \centering
  \includegraphics[width=\linewidth]{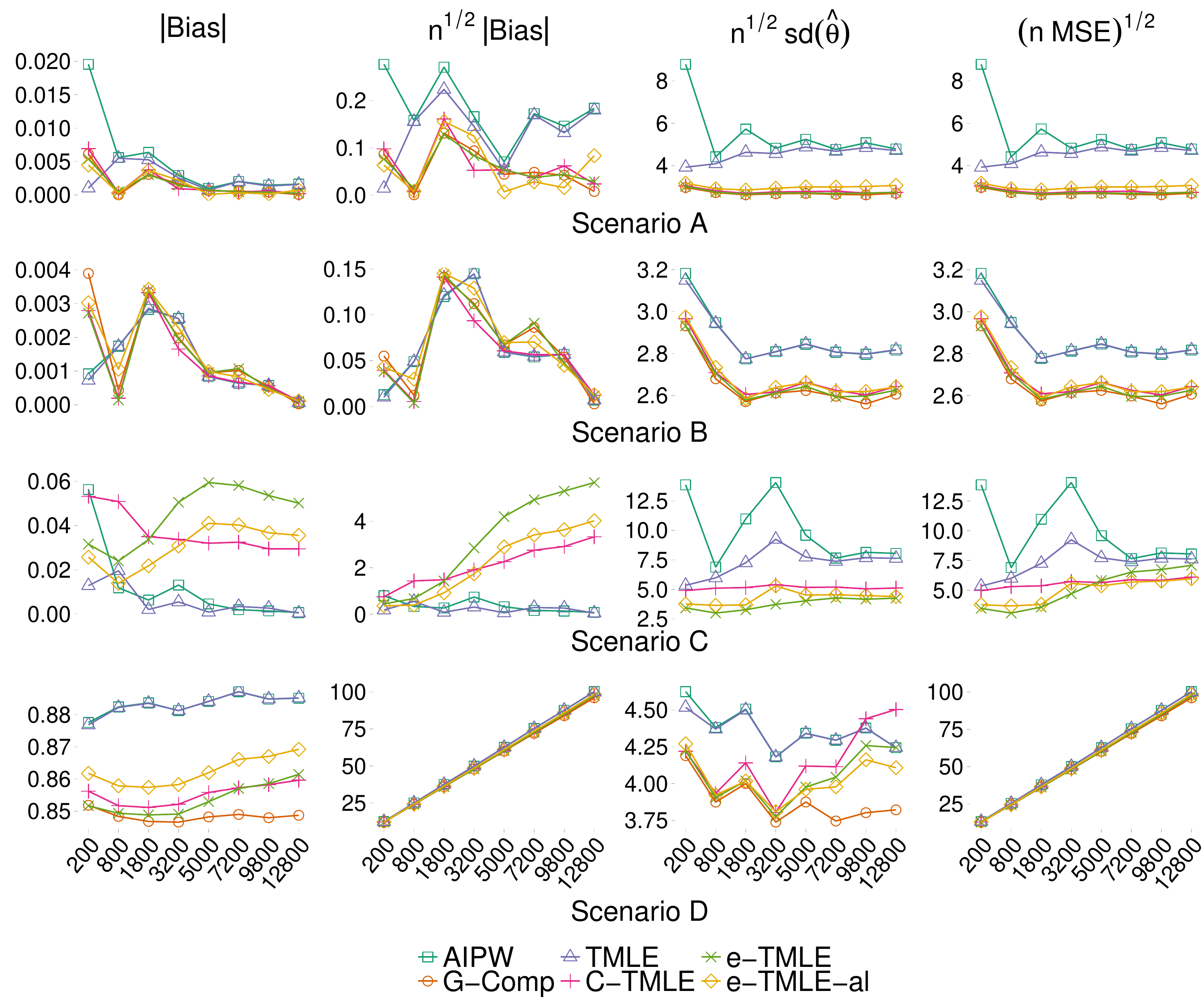}
  \caption{Simulation Results: bias, variance, and mean squared error
    for $\delta = 1$.}
  \label{fig:simula1}
\end{figure}

\begin{figure}[!htb]
  \centering
  \includegraphics[width=0.8\linewidth]{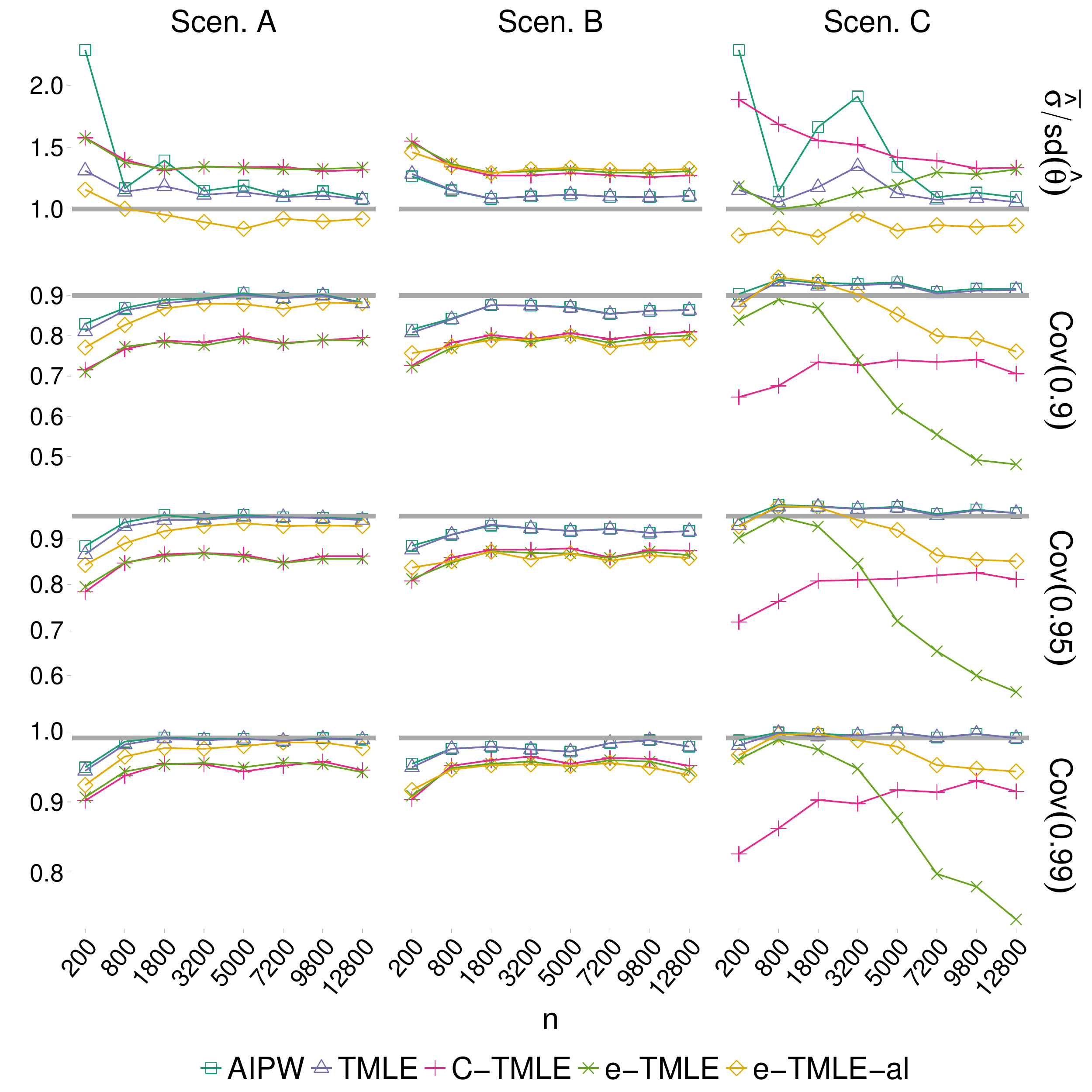}
  \caption{Simulation Results: variance estimation and coverage
    probabilities for $\delta = 1$.}
  \label{fig:simula2}
\end{figure}

\paragraph{Results for scenario A.}
The TMLE has smaller bias than all competitors in small samples
($n=200$). The AIPW and TMLE have similar asymptotic performance, with
the TMLE having much better small sample performance. This improvement
has been demonstrated in several simulation studies
\citep[e.g.,][]{Porter2011}. The variance of the TMLE is much larger
than the variance of its competitors (except AIPW), making its overall
performance on mean squared error worse. The $g$-computation estimator
and the C-TMLE have similar performance, with the $e$-TMLE having
comparable performance. Overall, the asymptotic efficiency gains
obtained with the C-TMLE and $e$-TMLE are noticeable, their MSE is
similar to that of the $g$-computation estimator, and much smaller
than that of the efficient estimator TMLE. In particular, it seems
that the $n^{1/2}$-bias of the TMLE dos not converge quickly enough,
perhaps as a result of the large variability of the inverse
probability weights. This problem is solved by the collaborative
double robustness involved in $e$-TMLE and C-TMLE, which are capable
of detecting that the outcome models are correctly specified, and
therefore do not adjust for the full covariate vector $W$ in the
propensity score. As predicted by Theorem~\ref{theo:dr}, the
confidence interval based on $e$-TMLE has asymptotically correct
coverage. This is also the case for the TMLE. However, the variance
for the C-TMLE based on the efficient influence function, which is the
default of the R package \texttt{ctmle} used in our simulations is
inconsistent and generates important undercoverage of the confidence
intervals. This is consistent with the simulation results reported in
\cite{ju2018collaborative}.

\paragraph{Results for scenario B.}  All estimators have similar
performance in terms of bias, with the TMLE having slightly smaller
bias at small sample sizes. The MSE of the C-TMLE and $e$-TMLE is
smaller at all sample sizes, but the difference is not as large as it
is for scenario (A). The MSE of all estimators is smaller than in
scenario (A), this is a consequence of the misspecification of the
propensity score model, which reduces the variability of the
estimator. The $e$-TMLE seems to be $n^{1/2}$-consistent in this
scenario, which is not predicted by our theory. According to our
asymptotic analysis in Section~\ref{sec:proposal},
$n^{1/2}$-consistency of $e$-TMLE requires that
$\beta(\tilde\lambda) = O_P(n^{-1/2})$, with $\beta(\tilde\lambda)$
defined in (\ref{eq:defbeta}). If $\hat m$ is the MLE in a correctly
specified parametric model for $m$, as in this simulation, then
$\beta(\tilde\lambda) = O_P(n^{-1/2})$ is expected. However,
$\beta(\tilde\lambda) = O_P(n^{-1/2})$ should not be expected in
general, for example for data-adaptive estimators $\hat m$. In this
scenario all confidence intervals have coverage probabilities below
the nominal level.

\paragraph{Results for scenario C.} In this scenario all estimators
had larger variance, compared to scenarios A and B. This is due to the
high variability of the propensity score weights. The $G$-computation
estimator had smaller bias than the C-TMLE and the $e$-TMLE, but this
is an artifact of our data generating mechanism and preliminary
estimator $\hat m$. Although the TMLE has smaller bias at all sample
sizes, the C-TMLE and $e$-TMLE have a better bias-variance trade-off
than the plan TMLE. In addition, the only estimator that seems to be
$n^{1/2}$-consistent is the TMLE. For the $e$-TMLE, this is a result
of (\ref{eq:defbeta}), which shows that in this case
$\beta(\tilde\lambda)$ is not $O_P(n^{-1/2})$. In this scenario the
confidence intervals for the $e$-TMLE and C-TMLE have coverage
probabilities below the nominal level. The interval based on the TMLE
has a coverage probability close to the nominal level. Under
consistent estimation of the propensity score, efficient estimation
theory predicts that this interval has conservative coverage. This is
corroborated in our simulations. The C-TMLE has better asymptotic
coverage than the $e$-TMLE-al in both cases $\delta=1$ and
$\delta=0$. We conjecture this is a consequence of the sparsity of our
data generating mechanisms, for which the C-TMLE is specially
designed.

\paragraph{Results for scenario D.} All estimators have similar bias
that does not disappear at $n^{1/2}$ rate, as predicted by theory. All
confidence intervals in this scenario have poor performance, and
are not shown in Figure~\ref{fig:simula2}.

\section{Illustrative Application}\label{sec:applica}

To illustrate our methods in a real dataset, we revisit the example
presented in \cite{kupzyk2017advanced,beal2014introduction}. The
dataset for this study is part of a longitudinal study of adolescent
girls originally conducted by \cite{dorn2008association}. We reanalyze
the data with the objective of assessing whether smoking among
adolescent girls negatively affects bone health via depletion in bone
mineral content (BMC). The main hypothesis we test is whether smoking
causes lower levels of accrual in BMC. Data on 259 adolescent girls
was collected each year for three years, and includes information on
smoking status, age, race, BMI, SES, age at menarche, Tanner breast
stage, birth control, calcium intake, PAQ-C physical activity score,
state anxiety T score, and trait anxiety T score. Bone mineral content
of the hip, spine, and total body was determined by dual-energy x-ray
absorptiometry. The challenge for causal inference is that smoking
behavior may be influenced by variety of reasons such as increased
depression and physical activity, and those factors may also affect
BMC (e.g., depression decreases BMC accrual, earlier onset of puberty
have higher BMC accrual levels, etc.) In this article, we will
estimate the effect of smoking on the first year of study on BMC
measured on the third year of the study. By the third year, 59 girls
were lost to follow up, so their outcome is missing. Denote
$T\in\{0,1\}$ an indicator of smoking status, and $M\in\{0,1\}$ an
indicator of not lost to follow-up. We will estimate the average
treatment effect by comparing the mean BMC in counterfactual worlds in
which $P(T=1,M=1)=1$ and $P(T=0,M=1)=1$, respectively. Denote with $W$
the confounders listed above. The propensity scores that we must
estimate are given by $P(T=t, M=1\mid W):t\in\{0,1\}$. We estimate
these probabilities by independently estimating $P(T=t\mid W)$ and
$P(M=1\mid T=t,W)$, and using the Bayes rule. In order to perform
model selection for the propensity score models and the outcome
regression, we use 5-fold cross-validation as implemented in the R
package SuperLearner \citep{SL}. We compared the risk (negative
log-likelihood for binary outcomes, MSE for continuous outcomes) of 9
different candidate prediction methods. The cross-validation results
are presented in Table~\ref{tab:psmodels}. We then computed the
estimators for the $e$-score separately for the two groups. The
original propensity score, together with the $e$-score, are presented
in Figure~\ref{fig:gaplica}. This figure illustrates the reduction in
variability of the propensity score achieved with the $e$-score. For
reference, the variance of the inverse propensity score for the
treated and untreated groups are 154.7 and 122.9, respectively. The
corresponding quantities for the $e$-score are $0.096$ and $0.0005$.

\begin{table}[!htb]
  \centering
  \begin{tabular}{r|rr|rr}
    \hline
    & \multicolumn{2}{c|}{Propensity score} &
                                             \multicolumn{2}{c}{Outcome
                                             regression} \\
    & Model for $T$ & Model for $M$ &
                                      Exposed
    & Controls \\
    \hline
    GLM                  & 0.4695 & 0.4868 & 3912.9 & 8676.7 \\
    GLM-$\ell_1$         & 0.4332 & 0.4783 & 3893.6 & 8047.9 \\
    GLM-$\ell_1\times 2$ & 0.4413 & 0.4850 & 3909.0 & 7542.8 \\
    GLM-$\ell_1\times 3$ & 0.4388 & 0.4820 & 3940.3 & 8046.9 \\
    Bayes GLM            & 0.4568 & 0.4810 & 3911.6 & 8681.6 \\
    MARS                & 0.6548 & 0.8024 & 5126.1 & 6852.7 \\
    GAM                 & 0.4747 & 0.4977 & 3956.7 & 7509.4 \\
    Boosted GLM          & 0.4334 & 0.4838 & 4021.7 & 8672.0 \\
    Boosted GAM         & 0.4377 & 0.4832 & 3868.9 & 9258.4 \\
    \hline
  \end{tabular}
  \caption{Cross-validated risk (negative log-likelihood for binary
    outcomes, MSE for continuous outcomes) for each
    estimator. GLM denotes generalized linear models with canonical link, GLM-$\ell_1$ denotes
    GLM with $\ell_1$ regularization, GLM-$\ell_1\times
    k$ denotes GLM with $k$-way interactions and
    $\ell_1$ regularization, Bayes GLM is Bayesian GLM
    with non-informative priors, MARS is multivariate adaptive splines, and GAM is
    generalized additive model. Tuning parameters were chosen using
    cross-validation with the caret R package \citep{Caret}.}
\label{tab:psmodels}
\end{table}

\begin{figure}[!htb]
  \begin{center}
    \includegraphics[scale = 0.2]{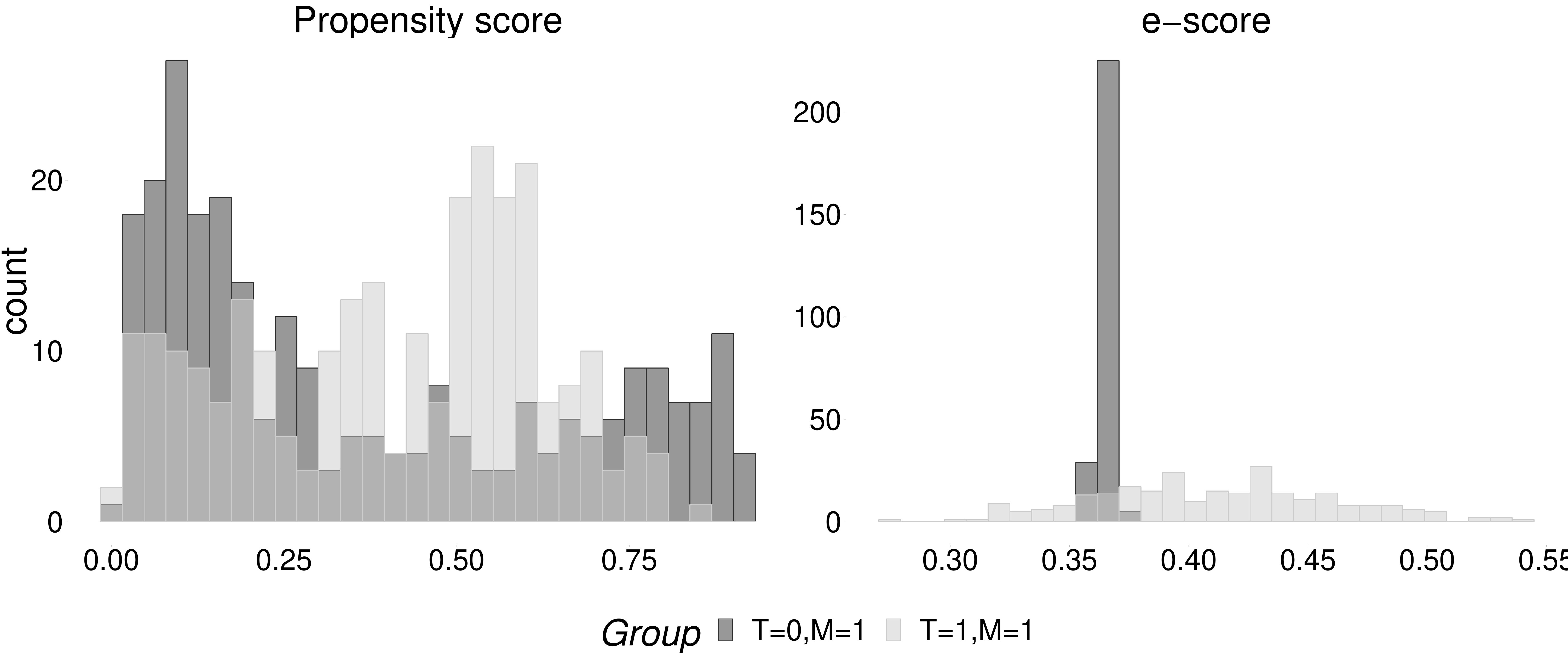}
  \end{center}
  \caption{Distributions of the estimated propensity and $e$- scores in the
    smoking--BMC example.}
  \label{fig:gaplica}
\end{figure}

We then proceeded to compute the estimators studied in
Section~\ref{sec:simula}. The C-TMLE was not computed because current
methodology and software does not allow for missingness in the
outcome. The results are presented in Table~\ref{tab:aplica}. The
point estimates are somewhat similar for all estimators, and show a
statistically significant reduction of around 40--45 grams of BMC due
to smoking. The standard error for the $e$-TMLE is not reported as we
do not have methodology to construct a valid estimate. As expected,
the AIPW and TMLE have very similar standard errors. The standard
error for the $e$-TMLE-AL is substantially different from the other
two estimates, yielding efficiency gains of approximately 54\%
compared to the TMLE. This means that a pre-specified analysis plan
using the $e$-TMLE-AL instead of the TMLE would have required roughly
$35\%\approx 1 - (1/1.54)$ fewer patients (90 out of 259) to achieve
the same power.

\begin{table}[ht]
  \centering
  \begin{tabular}{lccc}
    \hline
    Estimator & Estimate & S.E.    &  P-value   \\\hline
    AIPW      & -45.7    & 14.81 &  0.0020 \\
    TMLE      & -45.5    & 14.52 &  0.0017 \\
    $e$-TMLE   & -39.6    & --- &  --- \\
    $e$-TMLE-AL   & -43.3    & 11.67 &  0.0002 \\
    \hline
  \end{tabular}
  \caption{Estimates of the average treatment effect in the
    smoking--BMC example.}
  \label{tab:aplica}
\end{table}
\section{Discussion}\label{sec:discussion}

We have discussed several results for collaborative doubly robust
estimation of causal parameters. Our main contribution is the
introduction of the $e$-score, which greatly facilitates the
construction of collaborative doubly robust estimators compared to
existing methodologies that rely on model selection for the propensity
score or sparsity assumptions. Furthermore, we expect the introduction
of the $e$-score will facilitate the generalization of the methods to
more complex data structures, such as missing outcomes and
longitudinal studies.

A key component of the estimators that we propose is the estimation of
certain univariate regression functions using kernel regression. While
we propose to estimate the bandwidth of those estimators using the
optimal regression bandwidth, this choice may be suboptimal for
estimation of causal effects. Choosing the bandwidth in a way that
optimizes the MSE of the causal effect estimator will be the subject
of future research.

We have proved that one of our estimators, the $e$-TMLE-AL, is
asymptotically linear in the non-parametric model with a variance that
may be smaller than the variance of the efficient influence
function. Semiparametric efficiency theory \citep{Bickel97} dictates
that the variance of the efficient influence function is the smallest
possible variance attained by any regular estimator. Therefore, it
must be that the $e$-TMLE-AL is an irregular estimator, at least in
some cases. The implications of this irregularity are still unclear to
us. However, we hope the following argument may convince the reader
that the type of irregularity of this estimator is not very
important. Consider a non-parametric model $\cal M$ for the data
structure $(W_I,W_C,A,Y)$, where $W_I$ is an instrumental variable
(that is, a cause of $A$ but otherwise unrelated to $Y$), and $W_C$
are true confounders. Consider the efficiency bound $\tau^2$ for
estimation of $E(Y_1)$ in such a model. An efficient AIPW or TMLE
based on a data reduction given by $(W_C,A,Y)$ may have smaller
variance than the efficiency bound $\tau^2$, and therefore be
irregular in $\cal M$. However, this type of irregularity would hardly
seem problematic. Our conjecture is that our estimators rely on a
dimension reduction which has an effect similar to removing
instrumental variables, and therefore is irregular in the sense of the
above example. However, irregularity in this context remains poorly
understood, and rigorous mathematical methods need to be developed
before a rigorous study can be carried out.

Lastly, we conjecture that the $e$-score may have other important uses
in addition to the estimators discussed in this manuscript. One
particular use is as a tool for outcome model selection. According to
our Theorem~\ref{theo:cdr2}, an $e$-score with variance zero would
mean that the outcome regression, though possibly misspecified,
provides a $g$-computation consistent estimator of the causal
parameter of interest. A rigorous development of such methodology will
be the subject of future research.

\appendix 
\section{Supplementary Materials}
\subsection{Theorem~\ref{theo:cdr2}}
\begin{proof}
  For $\lambda=(m,e_1)$, we have $e_1D_{\lambda}=0$, and thus
  \begin{align*}
    e_1D_{\lambda_1}(O)&=e_1\{D_{\lambda_1}(O)-D_{\lambda}(O)\}\\
                       &=\int\frac{m(w) - m_1(w)}{e_1(w)}\{a-e_1(w)\}\dd \PP(a,w)
  \end{align*}
  The result for $m_1=m$ is trivially obtained from the above
  equation. The result for $g_1=g$ follows from the following
  argument.

  \begin{align}
    E_{\PP}D_{\lambda_1}(O)&=\int\frac{m(w) -
                             m_1(w)}{e_1(w)}\{a-e_1(w)\}\dd \PP(a,w)\notag\\
                           &=\int\frac{m(w) -
                             m_1(w)}{e_1(w)}\{g(w)-e_1(w)\}\dd \PP(w)\notag\\
                           &=\int\frac{a}{g(w)}\frac{y -
                             m_1(w)}{e_1(w)}\{g(w)-e_1(w)\}\dd \PP(y,a,w,g(w))\notag\\
                           &=\int\left[\int\frac{a}{g(w)}\frac{y -
                             m_1(w)}{e_1(w)}\{g(w)-e_1(w)\}\dd \PP(y,a,w\mid
                             g(w))\right]\dd\PP(g(w))\notag\\
                           &=\int\left[\int \frac{a}{g(w)}\{y -
                             m_1(w)\}\dd \PP(y,a,w\mid
                             g(w))\right]\frac{g(w)-e_1(w)}{e_1(w)}\dd\PP(g(w))\label{eq:p1}\\
                           &=\int\frac{r_1(w)}{e_1(w)}\{g(w)-e_1(w)\}\dd \PP(g(w))\label{eq:p2}\\
                           &=\int\frac{r_1(w)}{e_1(w)}\{g(w)-e_1(w)\}\dd \PP(w)\notag\\
                           &=\int\frac{r_1(w)}{e_1(w)}\{a-e_1(w)\}\dd \PP(a,w)\notag\\
                           &=\int\left[\int\frac{r_1(w)}{e_1(w)}\{a-e_1(w)\}\dd \PP(a,w,\mid
                             r_1(w))\right]\dd \PP(r_1(w))\notag\\
                           &=\int\frac{r_1(w)}{e_1(w)}\{e_1(w)-e_1(w)\}\dd \PP(r_1(w))\notag\\
                           &=0\notag,
  \end{align}
  where (\ref{eq:p1}) follows because $e_1(w)$ is a function of $w$ only
  through $g(w)$, and (\ref{eq:p2}) follows from the definition of
  $r_1(w)$.
\end{proof}


\subsection{Lemma~\ref{lemma:driftrep}}
\begin{proof}
  Define
  \begin{equation*}
    q^0(w) = e_1\left\{A\left(\frac{1}{\hat e(W)} -
        \frac{1}{\hat g(W)}\right)\,\bigg|\,\, m(W) = m(w)\right\},
  \end{equation*}
  where, as in $q$, the expectation is taken with respect to the
  distribution of $(A,W)$, taking $\hat e$ and $\hat g$ as fixed
  functions.  The proof proceeds as follows:
  \begin{align}
    \beta(\hat\lambda) =& \int \frac{1}{\hat e(w)}\{g(w) - \hat e(w)\}\{m(w) - \hat
                          m(w)\}\dd \PP(w)\notag\\
    =& \int
       g(w)\left\{\frac{1}{\hat e(w)} - \frac{1}{\hat g(w)}\right\}\{m(w)-\hat
       m(w)\}\dd \PP(w) \notag\\
                        &+ \int\frac{1}{\hat g(w)}\{g(w) - \hat g(w)\}\{m(w)-\hat m(w)\}\dd \PP(w)\notag\\
    =& \int
       g(w)\left\{\frac{1}{\hat e(w)} - \frac{1}{\hat g(w)}\right\}\{m(w)-\hat
       m(w)\}\dd \PP(w) + o_P(n^{-1/2})\label{sorder}\\
    =& \int
       a \left\{\frac{1}{\hat e(w)} - \frac{1}{\hat g(w)}\right\}\{m(w)-\hat
       m(w)\}\dd \PP(a, w) + o_P(n^{-1/2})\notag\\
    =& \int q^0(w) m(w) \dd \PP(w) - \int q(w) \hat m(w) \dd
       \PP(w)+ o_P(n^{-1/2})\label{defq}\\
    =& \int q(w) \{m(w) - \hat m(w)\} \dd \PP(w) - \int
       \{q(w) - q^0(w)\}m(w) \dd \PP(w)+ o_P(n^{-1/2})\notag\\
    =& \int q(w) \{m(w) - \hat m(w)\} \dd \PP(w) + o_P(n^{-1/2})\label{itexp}\\
    =& \int a \frac{q(w)}{g(w)} \{m(w) - \hat m(w)\} \dd
       \PP(a,w) + o_P(n^{-1/2})\notag\\
    =& \int a \frac{q(w)}{g(w)} \{y - \hat m(w)\} \dd \PP(y,a,w) + o_P(n^{-1/2})\notag,
  \end{align}
  where (\ref{sorder}) follows from Condition \ref{ass:DR2} in the
  main document of the manuscript, (\ref{defq}) follows from the
  definitions of $q(w)$ and $q^0(w)$, and (\ref{itexp}) follows
  from applying the law of iterated expectation to show that
  \[\int \{q(w) - q^0(w)\}m(w) \dd \PP(w)=0.\]
\end{proof}

\subsection{Theorem~\ref{theo:dr}}

Arguing as in equation (\ref{eq:wh}) of the main document we get
\begin{equation*}
  \etmle-\theta  = \beta(\tilde \lambda) +
  (\Pn - \PP)D_{\lambda, \theta} + o_P\big(n^{-1/2} +
  |\beta(\tilde\lambda)|\big)
\end{equation*}
Lemma~\ref{lemma:asbeta} below gives the asymptotic expression for
$\beta(\tilde \eta)$. Substituting this expression we get
\[\tmle-\theta = (\Pn - \PP)(D_{\lambda, \theta} - S(O)) +
  o_P\big(n^{-1/2} + O_P(n^{-1/2})\big),\]
where
\[S(O) = \frac{A}{\PP(A=1)}\frac{g_m(w) - \PP(A=1)}{g(W)}\{Y- m(W)\}.\]
This, together with the central
limit theorem concludes the proof.

\setcounter{lemma}{1}
\begin{lemma}[Asymptotic Linearity of $\beta(\tilde\lambda)$]\label{lemma:asbeta}
  In a slight abuse of notation, for a function $b:w\mapsto \mathbb{R}$, denote
  $K_{b,m}(O) = Ab(W)\{Y-m(W)\}$. Assume \ref{ass:donsker} and
  \ref{ass:DR2}. Then
  \begin{equation*}
    \beta(\tilde\lambda)= -(\Pn - \PP)K_{b, m}(O) + o_P(n^{-1/2}),
  \end{equation*}
  where \[b(w) = \frac{g_m(w) - \PP(A=1)}{\PP(A=1)g(W)}.\]
\end{lemma}
\begin{proof}From Lemma~\ref{lemma:driftrep}, we
  have
  \[\beta(\tilde\lambda)=\PP S_h(O) + o_P(n^{-1/2})=\PP K_{h, \hat m}(O) + o_P(n^{-1/2})\]
  By construction we have $\Pn K_{\hat h, \hat m}=0$, thus
  \begin{equation}
    \beta(\tilde\lambda)=-(\Pn- \PP)K_{\hat h, \hat m}-
    \PP(K_{h, m}-K_{h, \hat m}).\label{eq:betalin}
  \end{equation} We have
  \begin{align*}
    \PP(K_{h, m}-K_{h, \hat m}) &= \int a(h - \hat h)(y - \hat
                                  m)\dd \PP\\
                                &= \int a\left(\frac{q}{g} -
                                  \frac{q}{\hat g} +
                                  \frac{q}{\hat g} - \frac{\hat
                                  q}{\hat g}\right)(y - \hat
                                  m)\dd \PP\\
                                &= -\int \frac{q}{\hat g}(g -
                                  \hat g)(m - \hat m)\dd \PP + \int \frac{g}{\hat g}(q -
                                  \hat q)(m - \hat m)\dd \PP
  \end{align*}
  By assumption, we have
  \begin{equation} \PP(K_{h, m}-K_{h, \hat m})=O_P\big(||m - \hat m||\{||q
    - \hat q|| + ||g - \hat g||\}\big)=o_P(n^{-1/2}).\label{eq:rootn}
  \end{equation} Under Condition \ref{ass:DR1}, we have $\hat e(W)\to \PP(A=1)$, and
  therefore
  \[\hat q(w)\to E_{\PP}\left\{A\left(\frac{1}{\PP(A=1)} -
        \frac{1}{g(W)}\right)\,\bigg|\,\, m(W) = m(w)\right\}\]
  The law of total expectation shows that the right hand side of the
  above expression is equal to $g_m(w)/\PP(A=1) - 1.$ This shows
  that $\hat h\to b$. Application of Theorem 19.24 of
  \cite{vanderVaart98} to (\ref{eq:betalin}) shows that:
  \[    \beta(\tilde\lambda)= -(\Pn - \PP)K_{b, m}(O) +
    o_P(n^{-1/2}), \]
  completing the proof of the lemma.
\end{proof}

\subsection{Simulation Results}
\begin{figure}[H]
  \centering
  \includegraphics[width=\linewidth]{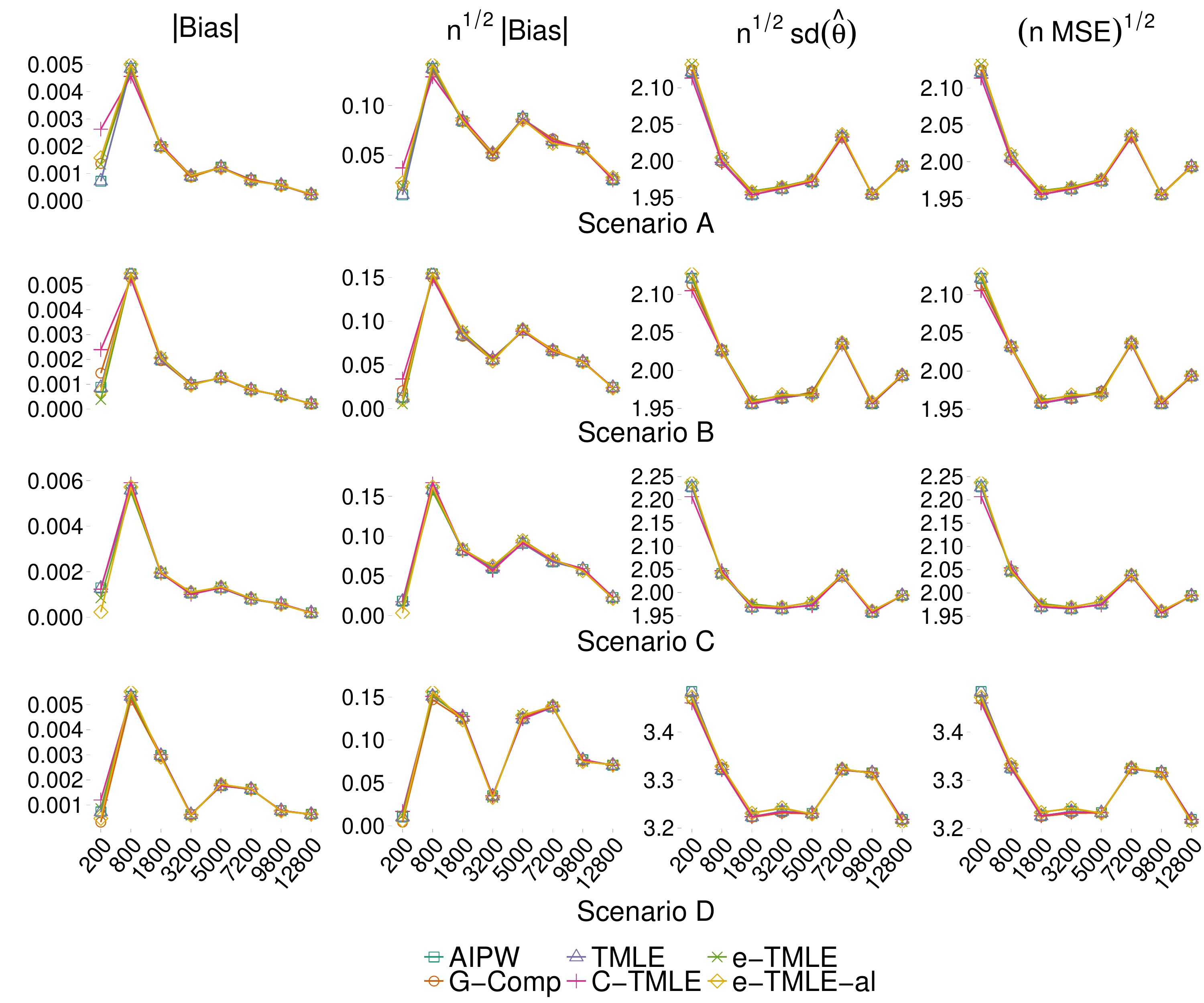}
  \caption{Simulation Results: bias, variance, and mean squared error
    for $\delta = 0$.}
  \label{fig:simula1}
\end{figure}

\begin{figure}[H]
  \centering
  \includegraphics[width=0.8\linewidth]{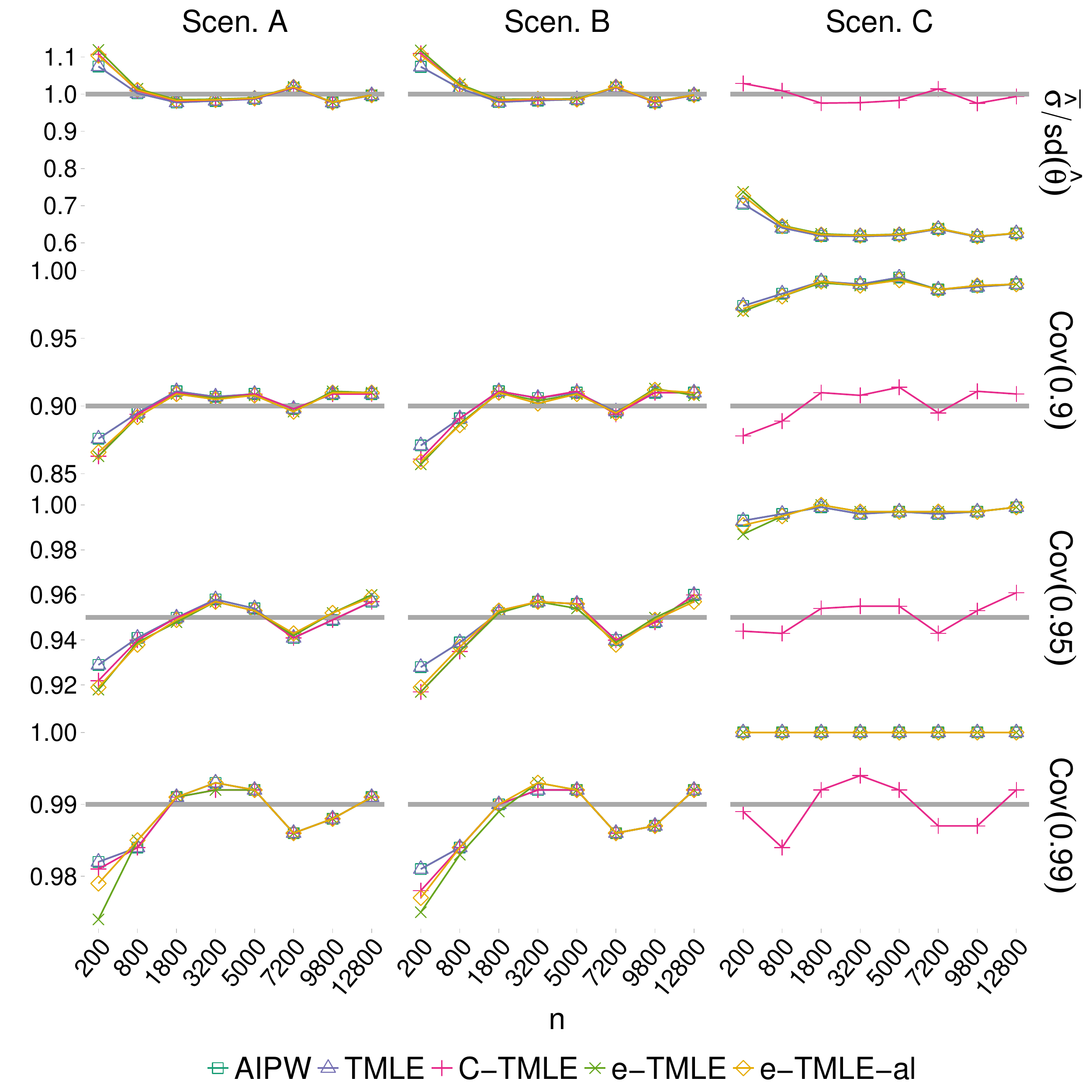}
  \caption{Simulation Results: variance estimation and coverage
    probabilities for $\delta = 0$.}
  \label{fig:simula2}
\end{figure}

\bibliographystyle{plainnat}
\bibliography{tmle}

\end{document}